\documentclass[11pt]{article}

\usepackage[a4paper,left=20mm,top=20mm,right=20mm,bottom=25mm]{geometry}

\usepackage{amsmath, amsfonts, amssymb, amsthm}
\usepackage[noend]{algorithmic}
\usepackage{authblk}

\usepackage{graphicx}  
\usepackage{mathtools}
\usepackage{subfig}
\usepackage[mathscr]{euscript}
\usepackage{enumitem}
\usepackage[normalem]{ulem}

\usepackage[colorlinks]{hyperref}
\hypersetup{
  citecolor=blue,
  linkcolor=blue,
}

\usepackage[font=footnotesize,width=.85\textwidth,labelfont=bf]{caption}

\usepackage{color}

\usepackage{wasysym}

\usepackage{authblk}
\usepackage[numbers, sort&compress]{natbib}

\usepackage[linesnumbered,ruled,vlined]{algorithm2e}

\newcommand{\AX}[1]{\textnormal{#1}}
\DeclareMathOperator{\lca}{lca}
\DeclareMathOperator{\child}{child}
\DeclareMathOperator{\outdegree}{outdeg}

\def\arrowedvec{\mathaccent"017E}
\newcommand{\G}{\arrowedvec{G}}
\newcommand{\rst}{_{|.}}
\newcommand{\hourglass}{\mathrel{\text{\ooalign{$\searrow$\cr$\nearrow$}}}}

\newtheorem{theorem}{Theorem}
\newtheorem{lemma}{Lemma}
\newtheorem{corollary}{Corollary}
\newtheorem{definition}{Definition}
\newtheorem{proposition}{Proposition}
\newtheorem{fact}[theorem]{Observation}

\newcommand{\cupdot}{\charfusion[\mathbin]{\cup}{\cdot}}
\makeatletter
\def\moverlay{\mathpalette\mov@rlay}
\def\mov@rlay#1#2{\leavevmode\vtop{%
    \baselineskip\z@skip \lineskiplimit-\maxdimen
    \ialign{\hfil$\m@th#1##$\hfil\cr#2\crcr}}}
\newcommand{\charfusion}[3][\mathord]{
  #1{\ifx#1\mathop\vphantom{#2}\fi
    \mathpalette\mov@rlay{#2\cr#3}
  }
  \ifx#1\mathop\expandafter\displaylimits\fi}
\DeclareRobustCommand\bigop[1]{%
  \mathop{\vphantom{\sum}\mathpalette\bigop@{#1}}\slimits@
}
\newcommand{\bigop@}[2]{%
  \vcenter{%
    \sbox\z@{$#1\sum$}%
    \hbox{\resizebox{\ifx#1\displaystyle.9\fi\dimexpr\ht\z@+\dp\z@}{!}{$\m@th#2$}}%
  }%
}
\makeatother
\title{Least resolved trees for two-colored best match graphs}

\author[1,2]{David Schaller}
\author[3]{Manuela Gei{\ss}}
\author[4]{Marc Hellmuth}
\author[1,2,5,6,7]{Peter F.\ Stadler}

\affil[1]{Max Planck Institute for Mathematics in the Sciences,
  Inselstra{\ss}e 22, D-04103 Leipzig, Germany
  \authorcr \texttt{sdavid@bioinf.uni-leipzig.de}}

\affil[2]{Bioinformatics Group, Department of Computer Science \&
  Interdisciplinary Center for Bioinformatics, Universit{\"a}t Leipzig,
  H{\"a}rtelstra{\ss}e~16--18, D-04107 Leipzig, Germany
  \authorcr \texttt{studla@bioinf.uni-leipzig.de}}

\affil[3]{Software Competence Center Hagenberg GmbH, Softwarepark 21,
  A-4232 Hagenberg, Austria
  \authorcr \texttt{manuela.geiss@scch.at}}

\affil[4]{Department of Mathematics, Faculty of Science, Stockholm University,
  SE - 106 91 Stockholm, Sweden
  \authorcr \texttt{marc.hellmuth@math.su.se}}

\affil[5]{Institute for Theoretical Chemistry, University of Vienna,
  W{\"a}hringerstra{\ss}e 17, A-1090 Wien, Austria}

\affil[6]{Facultad de
  Ciencias, Universidad Nacional de Colombia, Bogot{\'a}, Colombia}

\affil[7]{The Santa Fe Institute, 1399 Hyde Park Rd., Santa Fe, NM
  87501, USA}

\date{\ }

\begin{document}

\maketitle

%% --------------------------------------------------------------------
%       Abstract
%% --------------------------------------------------------------------

\abstract{2-colored best match graphs (2-BMGs) form a subclass of sink-free
  bi-transitive graphs that appears in phylogenetic combinatorics. There,
  2-BMGs describe evolutionarily most closely related genes between a pair
  of species. They are explained by a unique least resolved tree
  (LRT). Introducing the concept of support vertices we derive an
  $O(|V|+|E|\log^2|V|)$-time algorithm to recognize 2-BMGs and to construct
  its LRT. The approach can be extended to also recognize
  binary-explainable 2-BMGs with the same complexity. An empirical
  comparison emphasizes the efficiency of the new algorithm.}

\sloppy

\section{Introduction} 

Best match graphs recently have been introduced in phylogenetic
combinatorics to formalize the notion of a gene $y$ in species $2$ being an
evolutionary closest relative of a gene $x$ in species $1$, i.e., $y$ is a
best match for $x$ \cite{Geiss:19a}. The best matches between genes of two
species form a bipartite directed graph, the 2-colored best match graph or
2-BMG, that is determined by the phylogenetic tree describing the evolution
of the genes. 2-BMGs are characterized by four local properties
\cite{Geiss:19a,Korchmaros:20a} that relate them to previously studied
classes of digraphs:
\begin{definition}\label{def:bmg}
  A bipartite digraph $\G=(L,E)$ is a 2-BMG if it satisfies
  \begin{description}
  \item[\AX{(N0)}] Every vertex has at least one out-neighbor, i.e., $\G$
    is \emph{sink-free}.
  \item[\AX{(N1)}] If $u$ and $v$ are two independent vertices, then
    there exist no vertices $w$ and $t$ such that $(u,t), (v,w), (t,w) \in E$.
  \item[\AX{(N2)}] For any four vertices $u_1,u_2,v_1,v_2$ with
    $(u_1,v_1), (v_1,u_2), (u_2,v_2) \in E$ we have $(u_1,v_2)\in E$, i.e.,
    $\G$ is \emph{bi-transitive}.
  \item[\AX{(N3)}] For any two vertices $u$ and $v$ with a common
    out-neighbor, if there exists no vertex $w$ such that either
    $(u,w), (w,v) \in E$, or $(v,w), (w,u) \in E$, then $u$ and $v$ have
    the same in-neighbors and either all out-neighbors of $u$ are also
    out-neighbors of $v$ or all out-neighbors of $v$ are also out-neighbors
    of $u$.
  \end{description}
\end{definition}
Sink-free graphs have appeared in particular in the context of graph
semigroups \cite{Abrams:10} and graph orientation problems
\cite{Cohn:02}. Bi-transitive graphs were introduced in \cite{Das:20} in
the context of oriented bipartite graphs and investigated in more detail in
\cite{Korchmaros:20a,Korchmaros:20b}. The class of graphs satisfying
\AX{(N1)}, \AX{(N2)}, and \AX{(N3)} are characterized by a system of
forbidden induced subgraphs \cite{Schaller:20y}, see
Thm.~\ref{thm:2BMG-Fx-charac} below.

In general, best match graphs (BMGs) are defined as vertex-colored digraphs
$(\G,\sigma)$, where the vertex coloring $\sigma$ assigns to each gene $x$
the species $\sigma(x)$ in which it resides. The subgraphs of a BMG induced
by vertices of two distinct colors form a 2-BMG. Note that in this context
the vertex coloring is assigned \emph{a priori}, while Def.~\ref{def:bmg}
induces a coloring that is unique only up to relabeling of the colors
independently on each (weakly) connected component of $\G$. For each BMG
$(\G,\sigma)$, there is a unique least resolved leaf-colored tree
$(T^*,\sigma)$ with leaves corresponding to the vertices of $(\G,\sigma)$
such that the arcs in $(\G,\sigma)$ are the best matches w.r.t.\
$(T^*,\sigma)$ (cf.\ Def.\ \ref{def:bm} below).
Fig.~\ref{fig:LRT-2BMG-example} shows an example for a 2-BMG together with
its least resolved tree.  Using certain sets of rooted triples that can be
inferred from the 2-colored induced subgraphs of $(\G,\sigma)$ with three
vertices, it is possible to determine whether $(G,\sigma)$ is a BMG in
polynomial time and, if so, to construct the least resolved tree
$(T^*,\sigma)$ \cite{Geiss:19a,BMG-corrigendum}.  This work also describes
$O(|V|^3)$-time algorithms for the recognition of 2-BMGs and the
construction of the LRT for a given 2-BMG.

\begin{figure}[t]
  \begin{center}
    \includegraphics[width=0.8\textwidth]{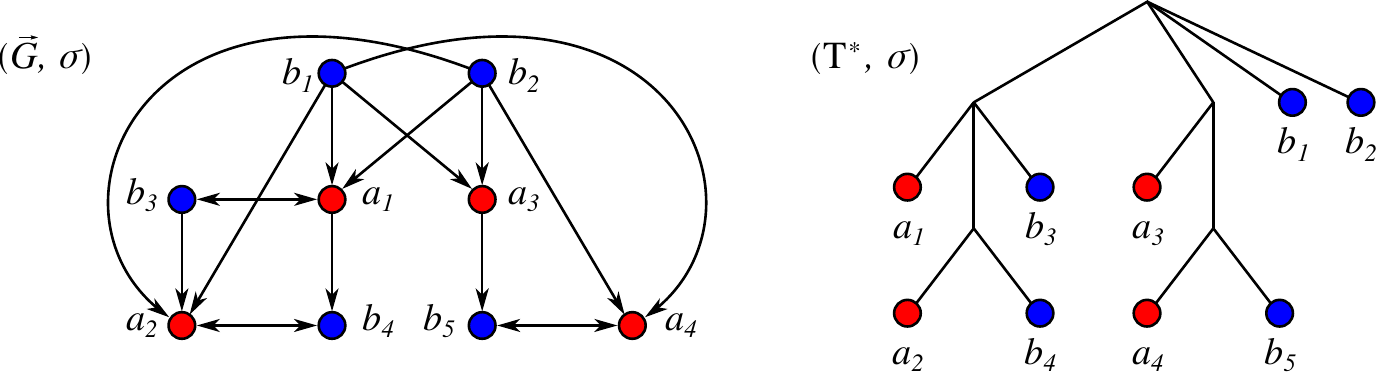}
  \end{center}
  \caption{Example for a 2-BMG $(\G,\sigma)$ and its explaining least
    resolved tree $(T^*,\sigma)$.}
  \label{fig:LRT-2BMG-example}
\end{figure}

In this contribution, we derive an alternative characterization of 2-BMGs
that avoids the use of rooted triples. This will give rise to an
alternative, efficient algorithm for the recognition of 2-BMGs and the
construction of the least resolved tree. The contribution is organized as
follows: In Sec.~\ref{sect:prelim}, we introduce the necessary notation and
review some results from the published literature that are needed later
on. Sec.~\ref{sect:LRT} is concerned with a more detailed analysis of the
least resolved trees (LRTs) of BMGs with an arbitrary number of colors.  We
then turn to the peculiar properties of the LRTs of 2-BMGs in
Sec.~\ref{sect:support}. To this end, we introduce the concept of ``support
leaves'' that uniquely determine the LRT. The main result of this section
is Thm.~\ref{thm:support-leaves-are-S}, which shows that the support leaves
of the root can be identified directly in the 2-BMG. In
Sec.~\ref{sect:algo}, we then turn Thm.~\ref{thm:support-leaves-are-S} into
an efficient algorithm for recognizing 2-BMGs and constructing their LRTs.
Computational experiments demonstrate the performance gain in practise.  In
Sec.~\ref{sect:be-2BMG} we extend the algorithmic approach to
binary-explainable 2-BMGs, a subclass that features an additional forbidden
induced subgraph.

\section{Preliminaries}
\label{sect:prelim}

Let $T=(V,E)$ be a tree with root $\rho$ and leaf set
$L\coloneqq L(T)\subset V$. The set of inner vertices of $T$ is
$V^0(T)\coloneqq V\setminus L$, in particular $\rho$ is an inner vertex. An
edge $e=uv\in E(T)$ is called an \emph{inner} edge of $T$ if $u$ and $v$
are both inner vertices. Otherwise it is called an \emph{outer} edge. We
consider leaf-colored trees $(T,\sigma)$ and write
$\sigma(L')\coloneqq \{\sigma(v)\mid v\in L'\}$ for subsets
$L'\subseteq L$.  A vertex $u\in V$ is an \emph{ancestor} of $v\in V$ in
$T$, in symbols $v\preceq_T u$, if $u$ lies on the path from $\rho$ to
$v$. For the edges $uv\in E(T)$ we use the convention that $uv\in E$,
$v\prec_T u$, $v$ is a child of $u$. We write $\child_T(u)$ for the set of
children of $u$ in $T$ and $T(u)$ for the subtree of $T$ rooted in $u$.
The \emph{least common ancestor} $\lca_{T}(A)$ is the unique
$\preceq_T$-smallest vertex that is an ancestor of all genes in
$A$. Writing $\lca_{T}(x,y)\coloneqq\lca_{T}(\{x,y\})$, we have
\begin{definition}\label{def:bm}
  Let $(T,\sigma)$ be a leaf-colored tree. A leaf $y\in L(T)$ is a
  \emph{best match} of the leaf $x\in L(T)$ if $\sigma(x)\neq\sigma(y)$ and
  $\lca(x,y)\preceq_T \lca(x,y')$ holds for all leaves $y'$ of color
  $\sigma(y')=\sigma(y)$.
  \label{def:BMG}
\end{definition}
Given $(T,\sigma)$, the graph $\G(T,\sigma) = (V,E)$ with vertex set
$V=L(T)$, vertex coloring $\sigma$, and with arcs $(x,y)\in E$ if and only
if $y$ is a best match of $x$ w.r.t.\ $(T,\sigma)$ is called the \emph{best
  match graph} (BMG) of $(T,\sigma)$ \cite{Geiss:19a}.
\begin{definition}\label{def:BestMatchGraph}
  An arbitrary vertex-colored graph $(\G,\sigma)$ is a \emph{best match
  graph (BMG)} if there exists a leaf-colored tree $(T,\sigma)$ such that
  $(\G,\sigma) = \G(T,\sigma)$. In this case, we say that $(T,\sigma)$
  \emph{explains} $(\G,\sigma)$.
\end{definition}

\begin{theorem}[\textnormal{\cite{Geiss:19a}, Thm.~9}]
  If $(\G,\sigma)$ is a BMG, then there is a unique least-resolved
  tree $(T,\sigma)$ that explains  $(\G,\sigma)$.
\end{theorem}

We say that $(\G,\sigma)$ is an $\ell$-BMG if $\sigma:V(\G)\to S$ is
surjective and $|S|=\ell$.  Given a directed graph $\G=(V,E)$ we denote the
set of out-neighbors of a vertex $x\in V$ by
$N(x)\coloneqq \{y\in V| (x,y)\in E(\G)\}$ and the out-degree $|N(x)|$ of
$x$ by $\outdegree(x)$. Similarly,
$N^{-}(x) \coloneqq \{y\in V| (y,x)\in E(\G)\}$ denotes the set of
in-neighbors. By construction, the coloring $\sigma$ of a BMGs
$(\G,\sigma)$ is \emph{proper}, i.e., $x\in N(y)$ implies
$\sigma(x)\ne\sigma(y)$, and there is at least one best match of $x$ for
every color $s\in\sigma(V)\setminus\{\sigma(x)\}$. In particular,
therefore, we have $N(x)\ne\emptyset$ for every 2-BMG, i.e., every 2-BMG is
sink-free. Note that BMGs will in general have sources, i.e., $N^{-}(x)$
may be empty. We write $\G[W]$ for the subgraph of $\G=(V,E)$ induced by
$W\subseteq V$ and $\G-W$ for $\G[V\setminus W]$.  A directed graph is
(weakly) connected if its underling undirected graph is connected.  A
\emph{connected component} is a maximal connected subgraph of $\G$.
  
Following \cite{Semple:03} we say that $T'$ is \emph{displayed} by $T$, in
symbols $T'\le T$, if the tree $T'$ can be obtained from a subtree of $T$
by contraction of edges. For leaf-colored trees we say that $(T,\sigma)$
\emph{displays} or \emph{is a refinement of} $(T',\sigma')$,
whenever $T'\le T$ and $\sigma(v)=\sigma'(v)$ for all $v\in L(T')$.
\begin{definition}
  An edge $e\in E(T)$ is redundant with respect to $\G(T,\sigma)$ if the
  tree $T_e$ obtained by contracting the edge $e$ satisfies
  $\G(T_e,\sigma)=\G(T,\sigma)$.
\end{definition}
We will need the following characterization of redundant edges:
\begin{lemma}[\textnormal{\cite{Schaller:20x}, Lemma~2.10}]
  \label{lem:redundant_edges}
  Let $(\G,\sigma)$ be a BMG explained by a tree $(T,\sigma)$.  The edge
  $e=uv$ in $(T,\sigma)$ is redundant w.r.t.\ $(\G,\sigma)$ if and only if
  (i) $e$ is an inner edge of $T$ and (ii) there is no arc $(a,b)\in E(\G)$
  such that $\lca_T(a,b)=v$ and
  $\sigma(b)\in \sigma(L(T(u))\setminus L(T(v)))$.
\end{lemma}

In the following we will frequently need the restriction of the coloring
$\sigma$ on $\G$ or $L(T)$ to a subset of vertices or leaves. Since in
situations like $(G_i,\sigma_{|V(G_i)})$ the set to which $\sigma$ is
restricted is clear, we will write $\sigma\rst$ to keep the notation less
cluttered.

BMGs can also be understood in terms of their connected components:
\begin{proposition}[\textnormal{\cite{Geiss:19a}, Prop.~1}]
  A digraph $(\G,\sigma)$ is an $\ell$-BMG if and only if all its connected
  components are $\ell$-BMGs.
  \label{prop:same-colorset}
\end{proposition}
As a simple consequence of Prop.\ \ref{prop:same-colorset} and by
definition of $\ell$-BMGs, all connected components $(G_i,\sigma\rst)$ and
$(G_j,\sigma\rst)$ of an $\ell$-BMG satisfy $\sigma(V(G_i))=\sigma(V(G_j))$
and $|\sigma(V(G_j))| = \ell$.  For our purposes it will also be important
to relate the structure of a tree $(T,\sigma)$ to the connectedness of the
BMG $\G(T,\sigma)$ that it explains.
\begin{proposition}[\textnormal{\cite{Geiss:19a}, Thm.~1}]
  \label{prop:bmg-connected}
  Let $(T,\sigma)$ be a leaf-labeled tree and $\G(T,\sigma)$ its BMG. Then
  $\G(T,\sigma)$ is connected if and only if there is a child $v$ of the root
  $\rho$ such that $\sigma(L(T(v)))\ne\sigma(L(T))$.
  Furthermore, if $\G(T,\sigma)$ is not connected, then for every connected
  component $\G_i$ of $\G(T,\sigma)$ there is a child $v$ of the root $\rho$ 
  such that $V(G_i)\subseteq L(T(v))$.
\end{proposition}

Moreover, 2-BMGs can be characterized by three types of forbidden subgraphs
\cite{Schaller:20y}. To this end we will need the following classes of
small bipartite graphs:
\begin{definition}[F1-, F2-, and F3-graphs]\par\noindent
  \begin{description}
  \item[\AX{(F1)}] A properly 2-colored graph on four distinct vertices
    $V=\{x_1,x_2,y_1,y_2\}$ with coloring
    $\sigma(x_1)=\sigma(x_2)\ne\sigma(y_1)=\sigma(y_2)$ is an
    \emph{F1-graph} if $(x_1,y_1),(y_2,x_2),(y_1,x_2)\in E$ and
    $(x_1,y_2),(y_2,x_1)\notin E$.
  \item[\AX{(F2)}] A properly 2-colored graph on four distinct vertices
    $V=\{x_1,x_2,y_1,y_2\}$ with coloring
    $\sigma(x_1)=\sigma(x_2)\ne\sigma(y_1)=\sigma(y_2)$ is an
    \emph{F2-graph} if $(x_1,y_1),(y_1,x_2),(x_2,y_2)\in E$ and
    $(x_1,y_2)\notin E$.
  \item[\AX{(F3)}] A properly 2-colored graph on five distinct vertices
    $V=\{x_1,x_2,y_1,y_2,y_3\}$ with coloring
    $\sigma(x_1)=\sigma(x_2)\ne\sigma(y_1)=\sigma(y_2)=\sigma(y_3)$ is an
    \emph{F3-graph} if\newline
    $(x_1,y_1),(x_2,y_2),(x_1,y_3),(x_2,y_3)\in E$ and
    $(x_1,y_2),(x_2,y_1)\notin E$.
  \end{description}
  \label{def:forbidden-subgraphs}
\end{definition}

\begin{theorem}[\textnormal{\cite{Schaller:20y}, Thm.~3.4}]
  \label{thm:2BMG-Fx-charac}
  A properly 2-colored graph is a 2-BMG if and only if it is sink-free and
  does not contain an induced F1-, F2-, or F3-graph.
\end{theorem}
As noted in \cite{Schaller:20y}, the forbidden induced F1-, F2-, and
F3-subgraphs characterize exactly the class of bipartite directed graphs
satisfying the Axioms \AX{(N1)}, \AX{(N2)}, and \AX{(N3)} mentioned in the
introduction.

Although we aim at avoiding the use of triples in the final results, we
will need them during our discussion. A triple $ab|c$ is a rooted tree $t$
on three pairwise distinct vertices $\{a,b,c\}$ such that
$\lca_{t}(a,b)\prec_t\lca_{t}(a,c)=\lca_{t}(b,c)=\rho$, where $\rho$
denotes the root of $t$. A set $\mathscr{R}$ of triples is
\emph{consistent} if there is a tree $T$ that displays all triples in
$\mathscr{R}$. Given a vertex-colored graph $(\G,\sigma)$, we define its
set of \emph{informative triples} \cite{Geiss:19a,Schaller:20x} as
\begin{equation}
  \mathscr{R}(\G,\sigma) \coloneqq
  \left\{ab|b' \colon
    \sigma(a)\neq\sigma(b)=\sigma(b'),\,
    (a,b)\in E(\G); %\text{ and }
    (a,b')\notin E(\G) \right\}.
  \label{eq:informative-triples}
\end{equation}

\begin{lemma}[\textnormal{\cite{Schaller:20x}, Lemma~2.8 and~2.9}]
  \label{lem:informative_triples}
  If $(\G,\sigma)$ is a BMG, then every tree $(T,\sigma)$ that explains
  $(\G,\sigma)$ displays all triples $t\in \mathscr{R}(\G,\sigma)$.\newline
  Moreover, if the triples $ab|b'$ and $cb'|b$ are informative for
  $(\G,\sigma)$, then every tree $(T,\sigma)$ that explains $(\G,\sigma)$
  contains two distinct children $v_1,v_2\in \child_T(\lca_T(a,c))$ such
  that $a,b\prec_T v_1$ and $b',c\prec_T v_2$.
\end{lemma}

\begin{fact}
  \label{obs:full-color-subtree}
  Let $(T,\sigma)$ be a tree explaining the BMG $(\G,\sigma)$, and
  $v\in V(T)$ a vertex such that $\sigma(L(T(v)))=\sigma(L(T))$.  Then
  $(a,b)\in E(\G)$ and $a\in L(T(v))$ implies $b\in L(T(v))$.
\end{fact}

Finally, there is a close connection between subtrees of $T$ and subgraphs
of $\G(T,\sigma)$. 
We have
\begin{lemma}[\textnormal{\cite{BMG-corrigendum}, Lemma~22 and~23}]
  \label{lem:subgraph}
  Let $(T,\sigma)$ be a tree explaining an BMG $(\G,\sigma)$.  Then
  $\G(T(u),\sigma\rst) = (\G[L(T(u))],\sigma\rst)$ holds for every
  $u\in V(T)$.  Moreover, if $(T,\sigma)$ is least resolved for
  $(\G,\sigma)$, then the subtree $T(u)$ is least resolved for
  $\G(T(u),\sigma\rst)$.
\end{lemma}

\section{Properties of Least Resolved Trees}
\label{sect:LRT}

In this short section we derive some helpful properties of LRTs which we
will use repeatedly throughout this work.
\begin{lemma}
  \label{lem:LRT-subtree-connected}
  Let $(\G,\sigma)$ be a BMG and $(T,\sigma)$ its least resolved tree.
  Then the BMG $\G(T(v), \sigma\rst)$ is connected for every $v\in V(T)$
  with $v\prec_{T}\rho_{T}$.
\end{lemma}
\begin{proof}
  By Lemma \ref{lem:subgraph}, $\G(T(v), \sigma\rst)$ is a BMG.  First
  observe that the BMG $\G(T(v), \sigma\rst)$ is trivially connected if $v$
  is a leaf.  Now let $v\prec_T \rho_{T}$ be an arbitrary inner vertex of
  $T$.  Thus, there exists a vertex $u\succ_T v$ such that $uv$ is an inner
  edge.  Since $(T,\sigma)$ is least resolved, it does not contain any
  redundant edges.  Hence, by contraposition of
  Lemma~\ref{lem:redundant_edges}, there is an arc $(a,b)\in E(\G)$ such
  that $\lca_T(a,b)=v$ and $\sigma(b)\in \sigma(L(T(u))\setminus L(T(v)))$.
  Since $a,b\in L(T(v))$, Lemma~\ref{lem:subgraph} implies that $(a,b)$ is
  also an arc in $\G(T(v), \sigma\rst)$.  Moreover, $\lca_{T(v)}(a,b)=v$
  clearly also holds in the subtree rooted at $v$.  Now consider the child
  $w\in\child_{T(v)}(v)$ such that $a\preceq_{T(v)} w$.  There cannot be a
  leaf $b'\in L(T(w))$ with $\sigma(b')=\sigma(b)$ since otherwise
  $\lca_{T(v)}(a,b')\preceq_{T(v)}w\prec_{T(v)}v$ would contradict that
  $(a,b)$ is an arc in $\G(T(v), \sigma\rst)$.  Thus
  $\sigma(b)\notin\sigma(L(T(w)))$. Since $\sigma(b)\in\sigma(L(T(v)))$, we
  thus conclude $\sigma(L(T(w)))\ne\sigma(L(T(v)))$.  The latter together
  with Prop.~\ref{prop:bmg-connected} implies that $\G(T(v), \sigma\rst)$
  is connected.
\end{proof}

The converse of Lemma~\ref{lem:LRT-subtree-connected}, however, is not
true, i.e., a tree $(T,\sigma)$ for which $\G(T(v), \sigma\rst)$ is
connected for every $v\in V(T)$ with $v\prec_{T}\rho_{T}$ is not
necessarily least resolved.  To see this, consider the caterpillar tree
$(T,\sigma)$ given by $(x'',(x',(x,y)))$ with
$\sigma(x)=\sigma(x')=\sigma(x'')\ne\sigma(y)$ and $u=\lca_T(x,x')$. It is
an easy task to verify that the BMG of each subtree of $T$ is connected.
However, the edge $\rho_T u$ is redundant.

\begin{lemma}
  \label{lem:single-color-is-leaf}
  Let $(T,\sigma)$ be the least resolved tree of some BMG $(\G,\sigma)$.
  Then every vertex $v\prec_T\rho_T$ with $|\sigma(L(T(v)))|=1$ is a leaf.
\end{lemma}
\begin{proof}
Let $v\prec_T\rho_T$ with $|\sigma(L(T(v)))|=1$ and assume, for
contradiction, that $v$ is not a leaf.
Hence, $|L(T(v))|>1$. By Lemma \ref{lem:subgraph} $\G(T(v), \sigma\rst)$
is a BMG and, therefore, properly colored. But then
$\G(T(v), \sigma\rst)$ is disconnected; a contradiction to Lemma 
\ref{lem:LRT-subtree-connected}.
\end{proof}
As a consequence we find
\begin{corollary}
  \label{cor:inner-vertex-mult-col}
  Let $(T,\sigma)$ be the least resolved tree of some BMG $(\G,\sigma)$.
  Then any vertex $v\in V(T)$ with $v\prec_T\rho_T$ is an inner vertex if
  and only if $|\sigma(L(T(v)))|>1$.
\end{corollary}
\begin{proof}
  If $|\sigma(L(T(v)))|=1$, Lemma \ref{lem:single-color-is-leaf} implies
  that $v$ is a leaf. Otherwise, if $|\sigma(L(T(v)))|>1$, $T(v)$ clearly
  must contain at least two leaves and thus $v$ cannot be a leaf.
\end{proof}

\section{Support Leaves} 
\label{sect:support}
In this section we introduce ``support leaves'' as a means to recursively
construct the LRT of a 2-BMG. The main result of this section shows that
these leaves can be inferred directly from the BMG without any further
knowledge of the corresponding LRT. We start with a technical result
similar to Cor.~3 in \cite{Geiss:19a}; here we use a much simpler, more
convenient notation.
\begin{lemma}\label{lem:support-leaves}
  Let $(T,\sigma)$ be the least resolved tree of a 2-colored BMG
  $(\G,\sigma)$.  Then, for every vertex $u\in V^0(T)\setminus \{\rho_T\}$,
  it holds $\child_T(u)\cap L(T)\ne\emptyset$. If $(\G,\sigma)$ is
  connected, then $\child_T(u)\cap L(T)\ne\emptyset$ holds for \emph{every}
  $u\in V^0(T)$.
\end{lemma}
\begin{proof}
  Suppose first that $(\G,\sigma)$ is disconnected and let
  $u\in V^0(T)\setminus \{\rho_T\}$.  Since $(T,\sigma)$ is least resolved,
  Lemma~\ref{lem:LRT-subtree-connected} implies that $\G(T(u), \sigma\rst)$
  is connected for every $u\in V(T)$ with $u\prec_T \rho_T$.  Hence, we can
  apply Prop.~\ref{prop:bmg-connected} to $\G(T(u), \sigma\rst)$ and
  conclude that there is a child $v\in \child_{T(u)}(u)$ such that
  $\sigma(L(T(v)))\ne\sigma(L(T(u)))$, hence in particular
  $\sigma(L(T(v)))\subsetneq\sigma(L(T(u)))$.  Since $(T,\sigma)$ is
  2-colored, the latter immediately implies $|\sigma(L(T(v)))|=1$ and, by
  Cor.~\ref{cor:inner-vertex-mult-col}, $v$ is a leaf.  Thus every
  $u\in V^0(T)\setminus\{\rho_T\}$ has a leaf $v$ among its children, i.e.\
  $\child_T(u)\cap L(T)\ne\emptyset$. If in addition $(\G,\sigma)$ is
  connected, we can apply the same argumentation to $u=\rho_T$ and conclude
  that a leaf $v$ is attached to $\rho_T$.
\end{proof}
Lemma~\ref{lem:support-leaves} states that, in the least resolved tree of a
connected 2-colored BMG, every inner vertex $u$ is adjacent to at least one
leaf, and thus in a way ``supported'' by it.  
\begin{definition}[Support Leaves]
  For a given tree $T$, the set $S_{u} \coloneqq \child_T(u)\cap L(T)$ is
  the set of all \emph{support leafs} of vertex $u\in V(T)$.
\end{definition}
Note that Lemma~\ref{lem:support-leaves} is in general not true for
$\ell$-BMGs with $\ell\ge3$, as exemplified by the (least-resolved) tree
$((a,b),(c,a'))$ with three distinct leaf colors
$\sigma(a)=\sigma(a')\neq \sigma(b)\neq \sigma(c)$.

As a simple consequence of Prop.~\ref{prop:bmg-connected} and
Cor.~\ref{cor:inner-vertex-mult-col}, we find
\begin{corollary}
  \label{cor:S_rho_empty}
  Let $(T,\sigma)$ be the least resolved tree (with root $\rho$) of some
  2-colored BMG $\G(T,\sigma)$. Then, $\G(T,\sigma)$ is connected if and
  only if $S_{\rho}\neq \emptyset$.
\end{corollary}
\begin{proof}
  By Prop.~\ref{prop:bmg-connected}, $\G(T,\sigma)$ is connected if and
  only if there exists a child $v$ of the root $\rho$ of $T$,
  $v\in\child_T(\rho)$, such that $T(v)$ does not contain all colors. Thus
  $|\sigma(L(T(v)))|=1$. By Cor.~\ref{cor:inner-vertex-mult-col}, we have
  $|\sigma(L(T(v)))|=1$ if and only if $v$ is a leaf, i.e.\ $v\in
  S_\rho$. Hence, $\G(T,\sigma)$ is connected if and only if
  $S_{\rho}\neq \emptyset$.
\end{proof}

\begin{lemma}
  \label{lem:remove-support-leaves}
  Let $(T,\sigma)$ be the least resolved tree of a 2-BMG $(\G,\sigma)$, and
  $S_{\rho}$ the set of support leaves of the root $\rho$. Then the
  connected components of $(\G-S_{\rho},\sigma\rst)$ are exactly the BMGs
  $\G(T(v),\sigma\rst)$ with $v\in\child(\rho)\setminus S_{\rho}$.
\end{lemma}
\begin{proof}
  Let $v\in\child_T(\rho)\cap V^0(T)=\child_T(\rho)\setminus S_{\rho}$ and
  consider the BMG $\G(T(v),\sigma\rst)$. By
  Lemma~\ref{lem:LRT-subtree-connected} and Lemma~\ref{lem:subgraph},
  $\G(T(v),\sigma\rst)$ is connected and we have
  $\G(T(v),\sigma\rst)=(\G[L(T(v))],\sigma\rst)$.  Moreover, it holds
  $((\G-S_{\rho})[L(T(v))],\sigma\rst)=(\G[L(T(v))],\sigma\rst)$ since
  $L(T(v)) = V(\G[L(T(v))]) = V(H[L(T(v))) ]$ for $H \coloneqq \G-S_{\rho} = 
  \G[V(\G)\setminus S_{\rho}]$.  

  If $\child_T(\rho)\setminus S_{\rho} = \{v\}$, then the statement is 
  trivially satisfied. Therefore, suppose that $|\child_T(\rho)\setminus 
  S_{\rho}|>1$.
  Hence, it remains to show that there are no arcs
  between $\G(T(v),\sigma\rst)$ and $\G(T(w),\sigma\rst)$ for any
  $w\in\child_T(\rho)\setminus S_{\rho}$, $w\neq v$.
  Cor.~\ref{cor:inner-vertex-mult-col} and $v\prec_T\rho$ imply that $T(v)$
  contains both colors.  Thus, by Obs.~\ref{obs:full-color-subtree}, there
  are no out-arcs to any vertex in $L(T)\setminus L(T(v))$, hence in
  particular there are no out-arcs $(x,y)$ with $x\preceq_T v$,
  $y\preceq_T w$. By symmetry, the same holds for $w$, thus we can conclude
  that there are no arcs $(y,x)$.  From the observation that each
  $x\in L(T)\setminus S_{\rho}$ must be located below some
  $v\in\child_T(\rho)\cap V^0(T)$, it now immediately follows that
  $(\G-S_{\rho},\sigma\rst)$ consists exactly of these connected components
  as stated.
\end{proof}

As a consequence, we have
\begin{corollary}
  \label{cor:lrt-recurse}
  Let $(T,\sigma)$ with root $\rho$ be the LRT of a 2-BMG $(\G,\sigma)$.
  Then each child of $\rho$ is either one of the support leaves $S_{\rho}$
  of $\rho$ or the root of the LRT for a connected component of
  $(\G-S_{\rho},\sigma\rst)$.
\end{corollary}
\begin{proof}
  Let $(T,\sigma)$ with root $\rho$ be the least resolved tree for
  $(\G,\sigma)$.  The support leaves $S_{\rho}$ are children of $\rho$ by
  definition.  By Lemma~\ref{lem:remove-support-leaves}, the connected
  components of $(\G-S_{\rho},\sigma\rst)$ are exactly the BMGs
  $\G(T(v),\sigma\rst)$ with $v\in\child_T(\rho)\setminus S_{\rho}$.
  Moreover, by Lemma~\ref{lem:subgraph}, the subtrees $T(v)$ with
  $v\in\child_T(\rho)\setminus S_{\rho}$ are exactly the unique LRTs for
  these BMGs.
\end{proof}

In order to use this property as a means of constructing the LRT in a
recursive manner, we need to identify the support leaves of the root
$S_{\rho}$ directly from the 2-BMG $(\G,\sigma)$ without constructing the
LRT first. To this end, we consider the set of \emph{umbrella vertices}
$U(\G,\sigma)$ comprising all vertices $x$ for which $N(x)$ consists of
\emph{all} vertices of $V(\G)$ that have the color distinct from
$\sigma(x)$. 
\begin{definition}[Umbrella Vertices]
For an arbitrary 2-colored graph $(\G,\sigma)$, the set 
\begin{equation*}
  U(\G,\sigma) \coloneqq  \left\{ x\in V(\G) \mid y\in N(x)
    \textrm{ if } \sigma(y)\ne \sigma(x) \textrm{ and } y\in V(\G) \right\}
\end{equation*}  
  is the set \emph{umbrella vertices} of  $(\G,\sigma)$.
\end{definition}
The intuition behind this definition is that every support leaf of
the root of the LRT of a 2-BMG must have all differently
colored vertices as out-neighbors, i.e., they are umbrella vertices.
We now define ``support sets'' of graphs as particular subsets
  of umbrella vertices. As we shall see later, support sets are closely
  related to support vertices in $S_{\rho}$.
\begin{definition}[Support Set of $(\G,\sigma)$]
  Let $(\G,\sigma)$ be a 2-colored graph. A \emph{support set}
  $S:=S(\G,\sigma)$ of $(\G,\sigma)$ is a maximal subset
  $S\subseteq U(\G,\sigma)$ of umbrella vertices such that $x\in S$ implies
  $N^-(x)\subseteq S$.
\end{definition}
\begin{lemma}
  Every 2-colored graph $(\G,\sigma)$ has a unique support set
  $S(\G,\sigma)$.
\end{lemma}
\begin{proof}
  Assume, for contradiction, that $(\G,\sigma)$ has (at least) two distinct
  support sets $S,S' \subseteq U(\G,\sigma)$.  Clearly neither of them can
  be a subset of the other, since supports sets are maximal.  We have
  $N^-(x)\subseteq S$ for all $x\in S$ and and $N^-(x')\subseteq S'$ for
  all $x'\in S'$, which implies that $N^-(z)\subseteq S\cup S'$ for all
  $z\in S\cup S'$. Together with the fact that $S$, $S'$, and thus
  $S\cup S'$, are all subsets of $U(\G,\sigma)$, this contradicts the
  maximality of both $S$ and $S'$.
\end{proof}

For the construction of the support set $S:=S(\G,\sigma)$, we
consider the following sequence of sets, defined recursively by
\begin{equation}
  S^{(k)}:= \{x\in S^{(k-1)}\mid N^-(x)\subseteq S^{(k-1)}\}
  \textrm{ for }k\ge 1 \textrm{ and } S^{(0)}=U(\G,\sigma).
\end{equation}
By construction $S^{(k+1)}\subseteq S^{(k)}$.  Furthermore, there is a
$k<|V(\G)|$ such that $S^{(k+1)}=S^{(k)}$.  Next we show that in a 2-BMG,
$S$ is obtained in a single iteration.
\begin{lemma}
  \label{lem:S-equals-Sprime}
  If $(\G,\sigma)$ is a 2-BMG, then $S=S^{(1)}$.
\end{lemma}
\begin{proof}
  Let $(\G=(V,E),\sigma)$ be a 2-BMG and $U=U(\G,\sigma)$. Assume for
  contradiction that $S\ne S^{(1)}$, and thus $S^{(2)}\subsetneq S^{(1)}$.
  We will show that this implies the existence of a forbidden F2-graph.  By
  assumption, there is a vertex $x_2\in S^{(1)}\setminus S^{(2)}$.  Thus,
  there must be a vertex $y_1\in N^-(x_2)$ (and thus $(y_1,x_2)\in E$) with
  $\sigma(y_1)\ne\sigma(x_2)$ such that $y_1\notin S^{(1)}$.  However, by
  definition, $y_1\in N^-(x_2)$ and $x_2\in S^{(1)}$ implies $y_1\in U$.
  Now, it follows from $y_1\in U\setminus S^{(1)}$ that there is a vertex
  $x_1\in N^-(y_1)$ with $\sigma(x_1)=\sigma(x_2)\ne\sigma(y_1)$ such that
  $x_1\notin U$.  The latter together with $x_2\in S^{(1)}\subseteq U$
  implies $x_1\ne x_2$. In particular, since $x_1\notin U$, the vertex
  $x_1$ does not have an out-arc to every differently colored vertex, thus
  there must be a vertex $y_2$ with $\sigma(y_2)=\sigma(y_1)$ such that
  $(x_1,y_2)\notin E$. Since $x_1\in N^-(y_1)$, we have $(x_1,y_1)\in E$
  and $y_1\neq y_2$.  Finally, $x_2\in U$ and
  $\sigma(y_2)=\sigma(y_1)\ne\sigma(x_2)$ implies that $(x_2,y_2)\in E$.
  In summary, we have four distinct vertices $x_1,x_2,y_1,y_2$ with
  $\sigma(x_1)=\sigma(x_2)\ne\sigma(y_1)=\sigma(y_2)$ and (non-)arcs
  $(x_1,y_1),(y_1,x_2),(x_2,y_2)\in E$ and $(x_1,y_2)\notin E$, and hence
  an induced F2-graph in $(\G,\sigma)$.  By Thm.~\ref{thm:2BMG-Fx-charac},
  we can conclude that $(\G,\sigma)$ is not a BMG; a contradiction.
\end{proof}

In general, $S=S^{(0)}=U(\G,\sigma)$ is not satisfied. To see this consider
the BMG $(\G,\sigma)$ that is explained by the triple $x_1 y|x_2$ with
$\sigma(x_1)=\sigma(x_2)\neq \sigma(y)$. One easily verifies that
$U(\G,\sigma)=\{x_1,x_2\}$ but $S=\{x_2\}$.

\begin{theorem}
  \label{thm:support-leaves-are-S}
  Let $(T,\sigma)$ be the least resolved tree of a 2-BMG
  $(\G,\sigma)$.  Then, the set of support leaves $S_{\rho}$ of the root
  $\rho$ equals the support set $S$ of $(\G,\sigma)$. In particular
  $S\ne\emptyset$ if and only if $(\G,\sigma)$ is connected.
\end{theorem}
\begin{proof}
  Let $(T,\sigma)$ be the LRT of a 2-BMG $(\G=(V,E),\sigma)$.
  We set $U\coloneqq U(\G,\sigma)$ and note first that $S=S^{(1)}$ by
  Lemma~\ref{lem:S-equals-Sprime}.

  First, suppose that $(\G,\sigma)$ is not connected. Then it immediately
  follows from Prop.~\ref{prop:bmg-connected} that
  $\sigma(L(T(v)))=\sigma(L(T))$ and thus $|\sigma(L(T(v)))|>1$ for any
  $v\in\child_T(\rho)$. The latter together with
  Cor.~\ref{cor:inner-vertex-mult-col} implies that any child of $\rho$
  must be an inner vertex in $T$. Hence, $S_{\rho}=\emptyset$. On the other
  hand, since $(\G,\sigma)$ is not connected, each of its connected 
  components is a 2-BMG (cf.\ Prop.~\ref{prop:same-colorset}),
  and thus, contains both colors.
  Therefore, for each vertex $x$ in $\G$, we
  can find a vertex $y$ with $\sigma(x)\neq\sigma(y)$ such that
  $(x,y),(y,x)\notin E$, and thus $x\notin S$. Since this is true for any
  vertex in $\G$, we can conclude $S=\emptyset=S_{\rho}$.

  Now, suppose that $(\G,\sigma)$ is connected. By
  Cor.~\ref{cor:S_rho_empty}, we have $S_{\rho}\neq \emptyset$. We first
  show $S_{\rho} \subseteq S$.  Let $x\in S_{\rho}$. By definition, $x$
  satisfies $\lca_{T}(x,y)=\rho$ and therefore $(x,y)\in E$ for all
  $y\in L(T)$ with $\sigma(y)\ne\sigma(x)$, i.e., $x$ has an out-arc to
  every differently colored vertex in $\G$.  By definition, we thus have
  $x\in U$.  Now assume for contradiction that
  $x\notin S=S^{(1)}= \{ z\in U \mid N^{-}(z)\subseteq U\}$.  The latter
  implies that there exists a vertex $y\in N^-(x)$ such that $y\notin U$.
  In particular, $(y,x)\in E$. Since $y\notin U$, there is some vertex $x'$
  with $\sigma(x')=\sigma(x)$ such that $(y,x')\notin E$.  Together this
  implies that $xy|x'$ is an informative triple. By
  Lemma~\ref{lem:informative_triples}, we obtain
  $\lca_T(x,y)\prec_T\lca_{T}(x,x')=\lca_{T}(x',y)\preceq_{T}\rho$; a
  contradiction to the assumption that $x$ is a support leaf of $\rho$.
  Thus $x\in S$.

  Next, we show by contraposition that $S\subseteq S_{\rho}$.  To this end,
  suppose that $x$ is not a support leaf of $\rho$, i.e.\
  $x\notin S_{\rho}$.  Hence, there is an inner vertex
  $v\in\child_T(\rho)\cap V^0(T)$ such that $x\prec_{T}v$.  By
  Cor.~\ref{cor:inner-vertex-mult-col}, we conclude that
  $|\sigma(L(T(v)))|=2$, i.e., the subtree $T(v)$ contains both colors.  We
  now distinguish two cases: (i) there is a leaf
  $y'\in L(T)\setminus L(T(v))$ with $\sigma(y')\ne\sigma(x)$, and (ii)
  there is no leaf $y'\in L(T)\setminus L(T(v))$ with
  $\sigma(y')\ne\sigma(x)$.

\smallskip
\textit{Case(i):} Since $T(v)$ contains both colors,
  there is a leaf $y\in L(T(v))$, with $y\ne y'$ and
  $\sigma(y)=\sigma(y')\ne\sigma(x)$.  Since, by construction, we have
  $\lca_{T}(x,y)\preceq_{T}v\prec_{T}\rho=\lca_{T}(x,y')$, it follows
  $(x,y')\notin E$.  Together with $\sigma(x)\ne\sigma(y')$, this
  immediately implies $x\notin U$. From $S^{(2)}\subseteq
  S^{(1)}\subseteq U$, we conclude $x\notin S^{(1)}=S$.

  \smallskip \textit{Case(ii):} Suppose that there is no leaf
  $y'\in L(T)\setminus L(T(v))$ with $\sigma(y')\ne\sigma(x)$.  We will
  continue by showing that there is a support leaf $y$ of vertex $v$ with
  $\sigma(y)\ne\sigma(x)$.  Assume, for contradiction, that the latter is
  not the case.  Since $(T,\sigma)$ is least resolved, the inner edge
  $\rho v$ is not redundant.  Hence, by Lemma~\ref{lem:redundant_edges},
  there must be an arc $(a,b)\in E$ such that $\lca_T(a,b)=v$ and
  $\sigma(b)\in \sigma(L(T)\setminus L(T(v)))$.  Since there is no leaf
  $y'\in L(T)\setminus L(T(v))$ with $\sigma(y')\ne\sigma(x)$, we conclude
  that $\sigma(b)=\sigma(x)$ and $\sigma(a)\ne\sigma(x)$.  Clearly, it
  holds $a,b\in L(T(v))$.  Now consider an arbitrary $a'\in L(T(v))$ with
  $\sigma(a')\ne\sigma(x)$.  Since, by assumption, every such $a'$ is not a
  support leaf of $v$, there must be an inner vertex $w\in\child_{T(v)}(v)$
  with $a'\prec_{T}w$.  By Cor.~\ref{cor:inner-vertex-mult-col} and since
  $w\prec_{T}v\prec_{T}\rho$, we conclude that $|\sigma(L(T(w)))|=2$, i.e.,
  the subtree $T(w)$ contains both colors.  Thus there is some $b'$ with
  $\sigma(b')=\sigma(x)$ and $\lca_T(a',b')\preceq_{T}w\prec_T v$.  Since
  $a'$ was chosen arbitrarily, we conclude that there cannot be an arc
  $(a,b)\in E$ such that $\lca_T(a,b)=v$; a contradiction.  It follows that
  there is a support leaf $y$ of vertex $v$ with $\sigma(y)\ne\sigma(x)$.
  Hence, $\lca_{T}(x,y)=v\preceq_{T}\lca_{T}(x'',y)$ for all $x''\in L(T)$
  with $\sigma(x'')=\sigma(x)$, and thus $(y,x)\in E$ and $y\in N^-(x)$.
  Since $S_{\rho}\ne\emptyset$ and
  $\sigma(y)\notin\sigma(L(T)\setminus L(T(v)))$, there must be a leaf
  $x'\in S_{\rho}$ with $\sigma(x')=\sigma(x)$.  The fact that
  $\lca_T(x,y)=v\prec_T\rho=\lca_T(x',y)$ implies $(y,x')\notin E$.
  Therefore and since $\sigma(x')\ne\sigma(y)$, it follows $y\notin U$.
  Together with $y\in N^-(x)$, we conclude that $x\notin S^{(1)}=S$.

  In summary, we have shown $S=S_{\rho}$ for any BMG $(\G,\sigma)$.
  Finally, $S=S_{\rho}$ together with Cor.\ \ref{cor:S_rho_empty} implies
  that $S\neq \emptyset$ if and only if $(\G,\sigma)$ is connected, which
  completes the proof.
\end{proof}

\section{Algorithmic Considerations} 
\label{sect:algo}

Thm.~\ref{thm:support-leaves-are-S} provides not only a convenient
necessary condition for connected 2-BMGs but also a fast way of determining
the support set $S=S_{\rho}$ and thus also a fast recursive approach to
construct the LRT for a 2-BMG. It is formalized in
Alg.~\ref{alg:2-col-BMG} and illustrated in
Fig.~\ref{fig:algo-2BMG-LRT-example}.

\begin{algorithm}
  \caption{LRT for connected 2-colored BMGs $(\G,\sigma)$.}
  \label{alg:2-col-BMG}
  \DontPrintSemicolon
  \SetKwFunction{FRecurs}{void FnRecursive}%
  \SetKwFunction{FRecurs}{Build2ColLRT}
  \SetKwProg{Fn}{Function}{}{}
  \KwIn{Connected properly 2-colored digraph $(\G=(L,E),\sigma)$, vertex
  $\rho$}
  \KwOut{LRT of $(\G,\sigma)$ if $(\G,\sigma)$ is a BMG}
  \Fn{\FRecurs{$\G, \sigma, \rho$}}{
     $U \leftarrow \left\{ x\in L \mid \outdegree(x) =
     |L|-|L[\sigma(x)]| \right\}$\label{line:start}\tcp*{umbrella vertices}
     $S^{(1)}\leftarrow \{ x\in U \mid N^{-}(x)\subseteq U\}$\tcp*{all
     in-neighbors in $U$}
     $S^{(2)}\leftarrow \{ x\in S^{(1)} \mid N^{-}(x)\subseteq S^{(1)}\}$
     \label{line:S}\tcp*{all in-neighbors in $S^{(1)}$}
     \uIf{$S^{(1)}=\emptyset$ or $S^{(2)}\ne S^{(1)}$\label{line:exit-cond}}%
     {\textbf{exit false}}
     \Else{
       \ForEach{$x\in S^{(2)}$}{\label{line:child-1}
         add $x$ as a child of $\rho$\label{line:child-2}\;
       }
       \ForEach{connected
       component $\G_v$ of $\G-S^{(2)}$\label{line:wcc}}{
         \If{$|V(\G_v)|=1$\label{line:exit-singleton}}{\textbf{exit false}}
         create vertex $v$\;
         $T_v\leftarrow$\FRecurs{$\G_v, \sigma\rst, v$}\label{line:rec_step}\;
         connect the root $v$ of $T_v$ as a child to $\rho$\label{line:end}\;
       }
     }
  }
\end{algorithm}

\begin{figure}[t]
  \begin{center}
    \includegraphics[width=0.85\textwidth]{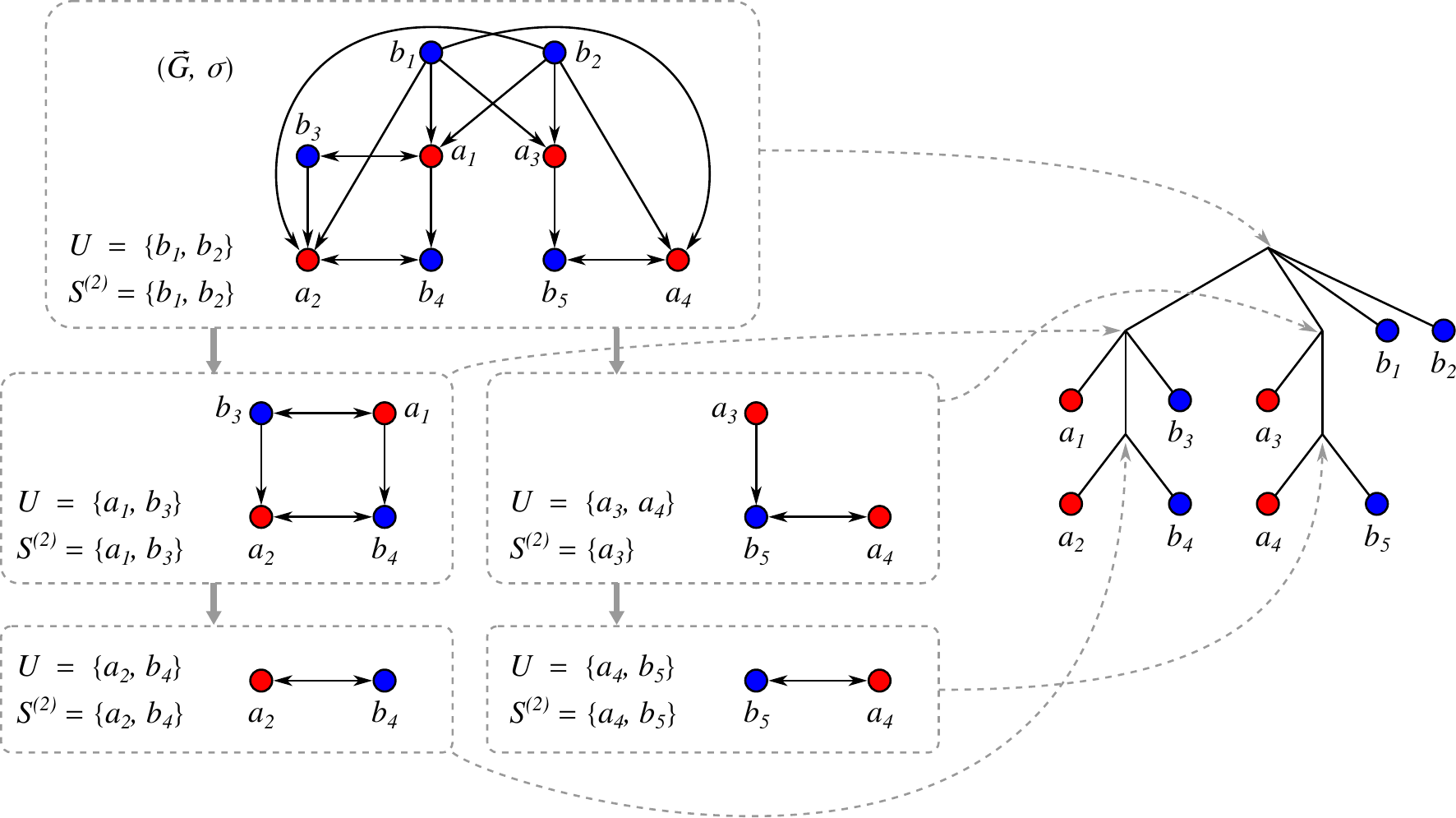}
  \end{center}
  \caption{Illustration of Alg.~\ref{alg:2-col-BMG} with input
    $(\G,\sigma)$ (uppermost box). The boxes indicate the five recursion
    steps that are necessary to decompose $(\G,\sigma)$, and correspond to
    the five inner vertices of the LRT shown on the right.  Note that, in
    the recursion step on $(\G[\{a_3,a_4,b_5\}], \sigma\rst)$, we have
    $U\ne S^{(2)}$.}
  \label{fig:algo-2BMG-LRT-example}
\end{figure}

\begin{lemma}
  \label{lem:build-2col-LRT}
  Let $(\G,\sigma)$ be a connected 2-BMG.
  Then Alg.~\ref{alg:2-col-BMG} returns the least resolved tree for
  $(\G,\sigma)$.
\end{lemma}
\begin{proof}
  Let $(T,\sigma)$ be the (unique) least resolved tree of $(\G,\sigma)$
  with root $\rho$.  The latter is supplied to Alg.~\ref{alg:2-col-BMG} to
  initialize the tree.  By Thm.~\ref{thm:support-leaves-are-S}, Lemma 
  \ref{lem:S-equals-Sprime} and since
  $(\G,\sigma)$ is connected, the set of support leaves
  $S_{\rho}=S^{(2)}=S^{(1)}\ne\emptyset$ for the root $\rho$ is correctly
  identified in the top-level recursion of Alg.~\ref{alg:2-col-BMG}
  (Line~\ref{line:start}-\ref{line:S}) and attached to the root $\rho$
  (Line~\ref{line:child-1}-\ref{line:child-2}). According to
  Cor.~\ref{cor:lrt-recurse}, one can now proceed to recursively construct
  the LRTs for the connected components of $(\G-S_{\rho},\sigma\rst)$,
  which is done in Line~\ref{line:wcc}-\ref{line:end}.  By
  Lemma~\ref{lem:remove-support-leaves}, these connected components
  $(\G_v,\sigma\rst)$ are exactly the BMGs $\G(T(v),\sigma\rst)$ with
  $v\in\child_T(\rho)\setminus\{S_\rho\}$ (Line~\ref{line:rec_step}).  In
  particular, therefore, we have $V(\G_v)=L(T(v))$. Since $v\notin S_\rho$,
  i.e., $v$ is an inner vertex, Cor.~\ref{cor:inner-vertex-mult-col} and
  $v\prec_T\rho$ imply $|\sigma(L(T(v)))|>1$.  Hence, in particular, the
  condition $|V(\G_v)|>1$ (cf.\ Line~\ref{line:exit-singleton}) to proceed
  recursively is satisfied for each connected component.  
\end{proof}

\begin{theorem}
  \label{thm:2BMG-algo-works}
  Given a connected properly 2-colored digraph $(\G,\sigma)$ as input,
  Alg.~\ref{alg:2-col-BMG} returns a tree $T$ if and only if $(\G,\sigma)$
  is a 2-colored BMG.  In particular, $T$ is the unique least resolved tree
  for $(\G,\sigma)$.
\end{theorem}
\begin{proof}
  By Lemma~\ref{lem:build-2col-LRT}, Alg.~\ref{alg:2-col-BMG} returns the
  unique least resolved tree $T$ if $(\G,\sigma)$ is a connected 2-colored
  BMG.  To prove the converse, suppose that Alg.~\ref{alg:2-col-BMG}
  returns a tree $T$ given the connected properly 2-colored digraph
  $(\G,\sigma)$ as input.  We will show that $(\G,\sigma)=\G(T,\sigma)$,
  and thus $(\G,\sigma)$ is a BMG.

  It is easy to see that $L(T)=V(\G)$ must hold since, in each step of
  Alg.~\ref{alg:2-col-BMG} every vertex is either attached to some inner
  vertex or passed down to a deeper-level recursion as part of some
  connected component.  Therefore, every vertex of $\G$ eventually appears
  in the output. Thus $\sigma(L(T))=\sigma(V(\G))$ and
  $|\sigma(L(T))|=|\sigma(V(\G))|=2$.  It remains to show
  $E(\G)=E(\G(T,\sigma))$.

  Note first that neither $(\G,\sigma)$ nor $\G(T,\sigma)$ contain arcs
  between vertices of the same color. Moreover, since
  Alg.~\ref{alg:2-col-BMG} eventually returns a tree, we have
  $S^{(1)}=S^{(2)}\ne\emptyset$ in every recursion step. Throughout the
  remainder of the proof, we will write $S^{(1)}_i$ and $S^{(2)}_i$ for the
  sets $S^{(1)}$ and $S^{(2)}$ of the $i^{th}$ recursion step.  Likewise,
  in every step, each connected component $(\G_v,\sigma\rst)$ computed in
  Line~\ref{line:wcc} must contain at least two vertices (cf.\
  Line~\ref{line:exit-singleton}), and thus $|\sigma(V(\G_v))|=2$ because
  $(\G,\sigma)$ is properly 2-colored.

  First, let $S$ be the support set of $\G(T,\sigma)$ and $x\in S$ be
  arbitrary.  Note that the support set is computed in the first iteration
  step of the algorithm as $S=S^{(2)}_1$, hence
  $S=S^{(2)}_1\ne\emptyset$. By construction of $T$, $x$ is attached as a
  leaf to $\rho$, i.e.\ $\lca_T(x,y)=\rho$. Consequently, $(x,y)$ is an arc
  in $\G(T,\sigma)$ for all $y\in V(\G)$ with $\sigma(y)\neq\sigma(x)$. By
  construction of $S$ in Alg.~\ref{alg:2-col-BMG}, we have
  $x\in S\subseteq U$, i.e.\ $x$ is an umbrella vertex in $(\G,\sigma)$ and
  has out-arcs to every vertex $y\in V(\G)$ with $\sigma(y)\ne\sigma(x)$.
  Hence, all arcs of the form $(x,y)$ with $x\in S$ and
  $\sigma(x)\neq\sigma(y)$ exist both in $(\G,\sigma)$ and in
  $\G(T,\sigma)$. The latter property is in particular satisfied for all
  vertices in $S$ and hence, all arcs between differently colored elements
  in $S$ exist both in $(\G,\sigma)$ and in $\G(T,\sigma)$.  Now consider
  an arbitrary vertex $y\in V(\G)\setminus S$.  Clearly, 
  all in-neighbors in $(\G,\sigma)$ of the elements in
  $S=S^{(2)}_1$ must be contained in $S$,
  as a consequence of the condition $S^{(1)}_1=S^{(2)}_1$ (cf.\
  Line~\ref{line:exit-cond}) and the construction of $S^{(1)}_1$ and
  $S^{(2)}_1$.  Hence, $y\notin S$ and $x\in S$ implies that $(y,x)$ is not
  an arc in $(\G,\sigma)$.  Moreover, $y\notin S$ also implies that $y$ is
  part of some connected component $(\G_v,\sigma\rst)$ of
  $(\G-S,\sigma\rst)$.  Therefore, and because Alg.~\ref{alg:2-col-BMG}
  returns $T$, we must have $y\in V(\G_v)= L(T(v))$ for some inner vertex
  $v\in\child_T(\rho)$.  As argued above, $(\G_v,\sigma\rst)$ and thus also
  the subtree $T(v)$ contain both colors. Together with
  Obs.~\ref{obs:full-color-subtree} and $x\notin L(T(v))$, this implies
  that $\G(T,\sigma)$ does not contain the arc $(y,x)$.  By the same
  arguments, there is no arc $(y,x')$ in $\G(T,\sigma)$ such that the
  vertex $x'$ is contained in a different connected component
  $(\G_{v'},\sigma\rst)\ne (\G_v,\sigma\rst)$ of $(\G-S,\sigma\rst)$ than
  $y$. Since $x\in S$ and $y\notin S$ were chosen arbitrarily, we conclude
  that (i) any arc incident to some vertex in $S$ exists in $(\G,\sigma)$
  if and only if it exists in $\G(T,\sigma)$, and (ii) $\G(T,\sigma)$
  contains no arcs between distinct connected components of
  $(\G-S,\sigma\rst)$.  Hence, it remains to consider the arcs within a
  connected component $(\G_v,\sigma\rst)$ of $(\G-S,\sigma\rst)$.

  Alg.~\ref{alg:2-col-BMG} recurses on each such connected component
  $(\G_v,\sigma\rst)$ using a newly created vertex $v\in\child_T(\rho)$ to
  initialize the tree $T(v)$.  By Lemma~\ref{lem:subgraph}, it clearly
  holds that, for any $x,y\in L(T(v))=V(\G_v)$, $(x,y)$ is an arc in
  $\G(T,\sigma)$ if and only it is an arc in $\G(T(v),\sigma)$.  Thus, it
  suffices to consider only the subtree $T(v)$.  Now, we can apply the same
  arguments as in the previous recursion step to conclude that all arcs
  incident to the support set $S^{(2)}_2$ constructed in the current
  recursion step are the same in $(\G,\sigma)$ and $\G(T,\sigma)$ and that
  neither $(\G,\sigma)$ nor $\G(T,\sigma)$ contain arcs between distinct
  connected components of $(\G_v-S^{(2)}_2,\sigma\rst)$. Hence, it suffices
  to consider the connected components of
  $(\G_v-S^{(2)}_2,\sigma\rst)$. Repeated application of this argumentation
  results in a chain of connected components that are contained in each
  other. Since Alg.~\ref{alg:2-col-BMG} finally returns a tree, this chain
  is finite, say with a last element $(\G_w-S^{(2)}_k,\sigma\rst)$, and
  thus $S^{(2)}_k=V(\G_w)$. In particular, therefore, every vertex in
  $V(\G)$ is contained in the support set of some recursion step.

  In summary, we have shown that $\G(T,\sigma)=(\G,\sigma)$.  Hence,
  $(\G,\sigma)$ is a connected 2-BMG and, by
  Lemma~\ref{lem:build-2col-LRT}, $T$ is the unique least resolved tree of
  $(\G,\sigma)$.
\end{proof}

The construction in Lines \ref{line:start}-\ref{line:S} in
Alg.~\ref{alg:2-col-BMG} naturally produces two cases, $U=S^{(1)}=S^{(2)}$
and $S^{(2)}\subseteq S^{(1)}\subsetneq U$.  The following result shows
that the latter case implies that the corresponding interior node in the
LRT has only a single non-leaf descendant:

\begin{lemma}
  Let $(\G,\sigma)$ be a 2-BMG and $S_\rho$ the support leaves of the root
  $\rho$ of its LRT $(T,\sigma)$.  If
  $W\coloneqq U(\G,\sigma)\setminus S_\rho \ne \emptyset$, then the
  following statements are true:
  \begin{enumerate}
  \item $S_\rho\ne\emptyset$, $\G$ is connected, and $\G-S_\rho$ is
    connected.
  \item All vertices in $U(\G,\sigma)=S_\rho \cupdot W$ have the same color,
  \item The set of support leaves $S_v$ of the unique inner vertex  
    child $v$ of $\rho$ contains vertices of both colors, 
    and 
  \item $W\subsetneq S_v$. 
  \end{enumerate}
  \label{lem:W}
\end{lemma}
\begin{proof}
  First recall that, by Thm.~\ref{thm:support-leaves-are-S} and the
  definition of the support set $S$ of $(\G,\sigma)$, we have
  $S_\rho=S\subseteq U(\G,\sigma)$, and thus
  $U(\G,\sigma)=S_\rho \cupdot W$.  Moreover, by
  Lemma~\ref{lem:remove-support-leaves}, the connected components of
  $(\G-S_{\rho},\sigma\rst)$ are exactly the BMGs $\G(T(v),\sigma\rst)$
  with $v\in\child(\rho)\setminus S_{\rho}$.  The vertices
  $v\in\child(\rho)\setminus S_{\rho}$ are all inner vertices of $T$ since,
  by definition, the support leaves $S_{\rho}$ are exactly the children of
  $\rho$ that are leaves.  Together with the contraposition of
  Lemma~\ref{lem:single-color-is-leaf} this implies that $T(v)$ contains
  both colors.
  
  \smallskip
  \noindent \textit{Statement 1:} Let $x\in W$, which exists due to the
  assumption $W\coloneqq U(\G,\sigma)\setminus S_\rho\ne\emptyset$.  Since
  $x\notin S_{\rho}$, it must be part of some connected component of
  $(\G-S_{\rho},\sigma\rst)$, say $\G(T(v),\sigma\rst)$ for some
  $v\in\child_T(\rho)\setminus S_{\rho}$.
  Now assume, for contradiction, that $\G-S_\rho$ consists of more than one
  connected component.  By Lemmas~\ref{lem:remove-support-leaves} and
  \ref{lem:single-color-is-leaf}, there is a vertex
  $v'\in\child_T(\rho)\setminus S_{\rho}$ such that $v\ne v'$ and both
    subtrees $T(v)$ and $T(v')$ contain both colors.  Hence, there are
  distinct $y\in L(T(v))$ and $y'\in L(T(v'))$ with
  $\sigma(y)=\sigma(y')\ne\sigma(x)$.  Together with $x\in L(T(v))$, we
  therefore have $\lca_T(x,y)\preceq_{T}v\prec_{T}\rho=\lca_{T}(x,y')$,
  which implies $(x,y')\notin E(\G)$.  However,
  $x\in W\subseteq U(\G,\sigma)$ and $\sigma(y')\ne\sigma(x)$ imply
  $(x,y')\in E(\G)$; a contradiction.  Hence, we conclude that $\G-S_\rho$
  has exactly one connected component, and thus $\rho$ has a single inner
  vertex child $v$.  Since $T$ is phylogenetic, the latter implies that
  $\rho$ must be incident to at least one leaf, i.e.\ $S_\rho\ne\emptyset$.
  Together with Thm.~\ref{thm:support-leaves-are-S} this in turn implies
  that $\G$ is connected.  In summary, Statement~1 is true.

  \smallskip
  \noindent \textit{Statement 2:} Let $x\in W$ as in the proof of Statement
  1. By arguments analogous to those used for Statement 1, we conclude that
  $\sigma(x)=\sigma(y)$ for every $y\in S_{\rho}$, since otherwise we would
  obtain $(x,y)\notin E(\G)$, and thus a contradiction to
  $x\in U(\G,\sigma)$.  Since $x\in W$ was chosen arbitrarily and
  $S_{\rho}$ is non-empty, we immediately obtain that all vertices in
  $U(\G,\sigma)=S_\rho \cupdot W$ have the same color, i.e., Statement~2 is
  true.
  
  \smallskip
  \noindent \textit{Statement 3:} Now consider the single inner vertex
  child $v$ of $\rho$, and its set of support leaves $S_v$, which must be
  non-empty by Lemma~\ref{lem:support-leaves}.  Note that $W$ must be
  entirely contained in $L(T(v))$ and recall that all vertices in
  $S_\rho \cupdot W$ are of the same color (cf.\ Statement 2).  First
  suppose, for contradiction, that $S_v$ only contains vertices of the
  \emph{opposite} color as the vertices in $S_\rho \cupdot W$. This
  immediately implies $S_v \cap W =\emptyset$, thus every vertex $x\in W$
  must be located in a subtree $T(w)$ of some inner vertex child $w$ of
  $v$. Again by contraposition of Lemma~\ref{lem:single-color-is-leaf},
  every such $T(w)$ contains both colors.  However, this contradicts
  $(x,y)\in E(\G)$ for every $y\in S_v$, which must hold as a consequence
  of $x\in W\subset U(\G,\sigma)$ and $\sigma(y)\ne\sigma(x)$.  Next
  suppose, for contradiction, that $S_v$ only contains vertices of the
  \emph{same} color as the vertices in $S_\rho \cupdot W$.  In this case,
  we obtain that the edge $\rho v$ is redundant w.r.t.\ $(\G,\sigma)$.  To
  see this, consider an arc $(x,y)\in E(\G)$ such that $\lca_{T}(x,y)=v$.
  Clearly, $x$ must be directly incident to $v$, since otherwise the
  subtree below $v$ to which $x$ belongs would contain both colors, and
  thus contradict $(x,y)\in E(\G)$. In other words, every such vertex $x$
  is a support leaf of $v$, thus $\sigma(x)=\sigma(S_v)=\sigma(S_\rho)$ and
  $\sigma(y)\neq\sigma(S_\rho)$. In particular, there exists no arc
  $(x,y)\in E(\G)$ such that $\lca_T(x,y)=v$ and
  $\sigma(y)\in \sigma(L(T)\setminus L(T(v)))=\sigma(S_\rho)$ and
  therefore, by Lemma~\ref{lem:redundant_edges}, the inner edge $\rho v$ is
  redundant. However, this contradicts the fact that $T$ is least resolved.
  In summary, only the case in which $S_v\ne\emptyset$ contains vertices of
  both colors is possible, and thus Statement~3 is true.
  
  \smallskip
  \noindent \textit{Statement 4:} First, recall from the proof of Statement
  3 that $W\subseteq L(T(v))$ for the single inner vertex child $v$ of
  $\rho$. In order to see that $W\subseteq S_v$, assume for contradiction
  that this is not the case. By similar arguments as used for showing
  Statement 3, this implies that some $x\in W$ lies in a 2-colored subtree
  $T(w)$ for some $w\in\child_T(v)\setminus S_v$. Together with the above
  established fact that $S_v$ contains both colors, this contradicts
  $x\in U(\G,\sigma)$.  Finally, $W\neq S_v$ is a consequence of the fact
  that $S_v$ contains both colors (Statement 3) but
  $W\subseteq S_\rho \cupdot W$ contains only one color (Statement 2).
\end{proof}

\begin{figure}[t]
  \begin{center}
    \includegraphics[width=0.85\textwidth]{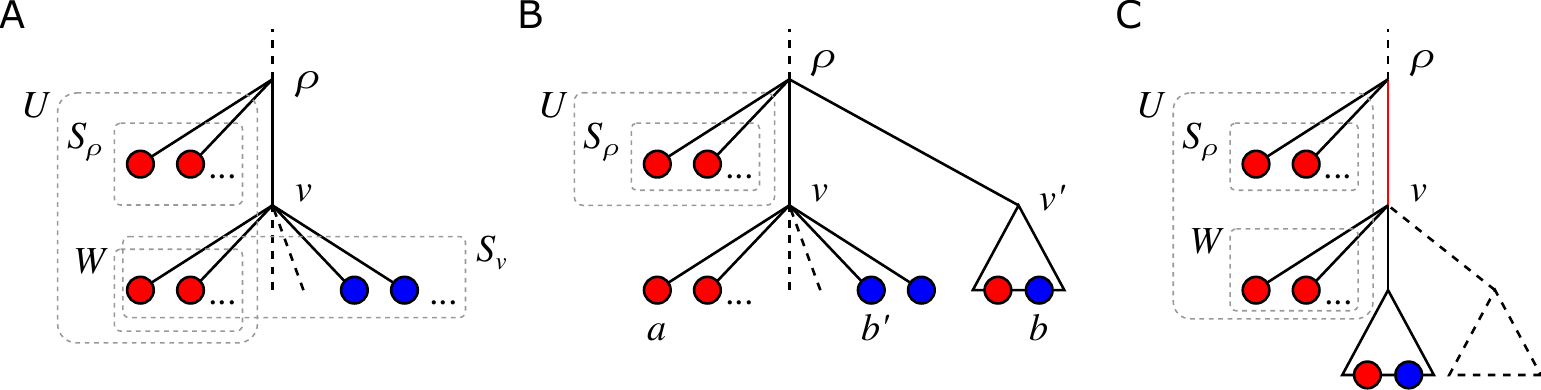}
  \end{center}
  \caption{Illustration of Lemma~\ref{lem:W}. (A) The (local) situation if
    $W=U\setminus S_{\rho}\ne\emptyset$ as implied by Lemma~\ref{lem:W}. In
    particular, $\rho$ only as a single inner vertex child $v$, all
    vertices in $U=S_{\rho}\cupdot W$ have the same color, $S_v$ contains
    vertices of both colors, and $W\subsetneq S_v$.  (B) There cannot be a
    second inner vertex child $v'$, since then none of the vertices except
    those in $S_{\rho}$ can be umbrella vertices, e.g.\ $(a,b)$ is not an
    arc in the graph explained by the tree in (B). Hence, this situation is
    not possible for $W\ne\emptyset$. (C) If $S_v$ does not contain
    vertices of both colors, then the edge $\rho v$ is redundant in the
    tree, contradicting that $(T,\sigma)$ in Lemma~\ref{lem:W} is the LRT.}
  \label{fig:U-and-S2-differ}
\end{figure}

We now use this result to consider the performance of
Alg.~\ref{alg:2-col-BMG}.
\begin{lemma}
  Alg.~\ref{alg:2-col-BMG} can be implemented to run in
  $O(|E|\log^2|V|)$ time for a connected input graph.
  \label{lem:connG}
\end{lemma}
\begin{proof}
  Since $\G$ is connected by assumption, we have $|V|\in O(|E|)$.  Starting
  from $(\G,\sigma)$, the list of out-degrees can be constructed in
  $O(|E|)$. The initial umbrella set $U$ is then obtained by listing
  the vertices with maximal out-degree in the color class. The initial
  set $S^{(1)}$ is constructed by checking, for each $u\in U$, the
  in-neighbors of $u$ for membership in $U$ in $O(|V|+|E|)$ operations.
  Then $S^{(2)}$ is obtained in the same manner from $S^{(1)}$,
  requiring $O(|V|+|E|)$ operations. The initial umbrella set $U$ and the
  sets $S^{(1)}$ and $S^{(2)}$ thus can be constructed in linear time.
  In each recursive call of \texttt{Build2ColLRT}, at least one leaf is
  split off, hence the recursion depth is $|V|-1$ in the worst case. Since
  the support vertices removed in each step have all of their in-neighbors
  in $U$, their removal does not affect the out-neighborhood for any
  $x\in V(\G-U)\subseteq V(\G-S^{(2)})$, and hence, $\outdegree(x)$ does
  not require updates. The in-neighborhoods $N^-(x)$ can be updated by
  removing the arcs between $\G-S^{(2)}$ and $S^{(2)}$ as a consequence of
  Lemma~\ref{lem:remove-support-leaves} and
  Thm.~\ref{thm:support-leaves-are-S}. Since every arc appears exactly once
  in the removal, the total effort for these updates is $O(|E|)$.

  We continue by showing that every vertex needs to be considered as an
  umbrella vertex at most twice, and that the total effort of constructing
  all sets $S^{(1)}$ and $S^{(2)}$ is $O(|E|)$, given that the umbrella
  vertices $U$ can be obtained efficiently, which we discuss afterwards. To
  this end, we distinguish, for each of the single recursion steps, two
  cases: $S^{(1)}=U$ and $S^{(1)} \subsetneq U$.  First if $S^{(1)}=U$, and
  thus also $S^{(2)}=S^{(1)}=U$, we consider each in-arc of $x\in U$. Since
  these vertices and their corresponding arcs are removed when constructing
  $\G-S^{(2)}$, they are not considered again in a deeper recursion
  step. In the second case, we have $S^{(1)} \subsetneq U$, which together
  with $S^{(2)}=S^{(1)}$ implies
  $W\coloneqq U\setminus S^{(2)} \ne \emptyset$, and only the vertices in
  $U\setminus W$ are removed. However, Lemma~\ref{lem:W} guarantees that,
  for a 2-BMG as input graph, the vertices in $W$ appear as support leaves
  in the next step and thus appear in the update of $U$, $S^{(1)}$, and
  $S^{(2)}$ no more than a second time.  In order to use the properties in
  Lemma~\ref{lem:W} for the general case (i.e.\ $(\G,\sigma)$ is not
  necessarily a BMG), we can, whenever $W\ne\emptyset$, (i) check that
  $\G-S^{(2)}$ only has a single connected component $\G_v$, and (ii) pass
  down the set $W$ to the recursion step on $\G_v$ in which the condition
  $W\subsetneq S^{(2)}$ is checked. If any of these checks fails, we can
  exit false. This way, we ensure that every vertex appears at most two
  times as an umbrella vertex in the general case.  To construct $S^{(1)}$
  from $U$, we have to scan the in-neighborhood $N^-(x)$ of each vertex
  $x\in U$ and check whether $N^-(x)\subset U$. We repeat this step to
  construct $S^{(2)}$ from $S^{(1)}$. Membership in $U$ and $S^{(1)}$,
  resp., can be checked in constant time (e.g.\ by marking the vertices in
  the current set $U$). Since we have to consider each vertex, and hence,
  each in-neighborhood at most twice, all sets $S^{(1)}$ and $S^{(2)}$ can
  be obtained with a total effort of $O(|E|)$.
  
  It remains to show that the input graph can be decomposed efficiently in
  such a way that the connectivity information is maintained and the
  candidates for umbrella vertices in each component are updated. The
  connected components $\G_v$ can be obtained by using the dynamic data
  structure described in \cite{Holm:01}, often called HDT data structure.
  It maintains a maximal spanning forest representing the underlying
  undirected graph with edge set
  $\widetilde{E}=\{xy \mid (x,y)\in E \textrm{ or } (y,x)\in E\}$, and
  allows deletion of all $|\widetilde E|\in O(|E|)$ edges with amortized
  cost $O(\log^2|V|)$ per edge deletion. The explicit traversal of the
  connected components to compute $U$ can be avoided as follows: Since
  $\outdegree(x)$ does not require updates, we can maintain a doubly-linked
  list of vertices $x$ for each color $i\in\{1,2\}$, and each value of
  $\outdegree(x)$ where $\sigma(x)=i$.  In order to be able to access the
  highest value of the out-degrees, we maintain these values together with
  pointer to the respective doubly-linked list in balanced binary search
  trees (BST), one for each color and each connected component.  The BSTs
  for the two colors are computed first for $(\G,\sigma)$ in
  $O(|V|\log(|V|))$ time and afterwards updated to fit with the out-degree
  of the currently considered component $\G_v$. To update these lists and
  BSTs for $\G_v$, observe first that $\G_v$ can be obtained from $G$ by
  stepwise deletion of single arcs, i.e.\ edges in the HDT data structure
  representing the underlying undirected versions.  We update, resp.,
  construct the pair of BSTs (one for each color) for each connected
  component as follows: Since a single arc deletion splits a connected
  component $\G'$ into at most two connected components $\G_1$, and $\G_2$,
  we can apply the well-known technique of traversing the smaller component
  \cite{Shiloach:81}.  The size of each connected component can be queried
  in $O(1)$ time in the HDT data structure.  Suppose w.l.o.g.\ that
  $|V(\G_1)|\leq |V(\G_2)|$. We construct a new pair of BSTs for $\G_1$,
  and delete the vertices $V(\G_1)$ and the respective degrees from the two
  original BSTs for $\G$, which then become the BSTs for $\G_2$.  More
  precisely, we delete each vertex $x\in V(\G_1)$ in the respective list
  corresponding to $\outdegree(x)$, and if the length of this list drops to
  zero, we also remove the corresponding out-degree in the BST. Likewise,
  we insert the out-degree of $x$ and an empty doubly-linked list into the
  newly-created BST for $\G_1$, if it is not yet present, and append $x$ to
  this list.  Note that the number of out-degree deletions and insertions
  does not exceed $|V(\G_1)|$.  Due to the technique of traversing the
  smaller component, every vertex is deleted and inserted at most
  $\lfloor\log |V| \rfloor$ times.  Therefore, we obtain an overall
  complexity of $O(|V|\log^2 |V|)$ for the maintenance of the BSTs where
  the additional log-factor originates from rebalancing the BSTs whenever
  necessary.

  In each recursion step, the set $U$ can now be obtained by listing (at
  most) the vertices with the maximal out-degree for each of the two
  colors. Finding the two out-degrees and corresponding lists in the BSTs
  requires $O(\log |V|)$ in each step, and thus $O(|V| \log |V|)$ in total.
  In order to determine whether these candidates $x$ are actually umbrella
  vertices, we have to check whether
  $\outdegree(x) = |V(G_v)|-|V(G_v)[\sigma(x)]|$.  The HDT data structure
  allows constant-time query of the size of a given connected component,
  since this information gets updated during the maintenance of the
  spanning forest. By the same means, we can keep track of the number of
  vertices of a specific color in each connected components.  Note that we
  only need to do this for one color $r$ since
  $|V(G_v)[s]|=|V(G_v)|-|V(G_v)[r]|$. This does not increase the overall
  effort for maintaining the data structure since it happens alongside the
  update of $|V(G_v)|$.

  In summary, the total effort is dominated by maintaining the
  connectedness information while deleting $O(|E|)$ arcs, i.e.,
  $O(|E|\log^2 |V|)$ time.
\end{proof}

As a direct consequence of Thm.~\ref{thm:support-leaves-are-S} the LRT of a
disconnected graph $\G$ is obtained by connecting the roots of the LRTs
$T_v$ of the connected components $G_v$ to an additional root vertex, see
also \cite[Cor.~4]{Geiss:19a}. Lemma~\ref{lem:connG} thus implies
\begin{theorem}
  \label{thm:2BMG-general-compl}
  The LRT of a 2-BMG can be computed in $O(|V| + |E|\log^2|V|)$. 
\end{theorem}
\begin{proof}
  The connected components $G_i=(V_i,E_i)$ of $G=(V,E)$ can be enumerated
  in $O(|V|+|E|)$ operations, e.g.\ using a breadth-first search on the
  underlying undirected graph. By Lemma~\ref{lem:connG},
  $O(|E_i|\log^2|V_i|)\le O(|E_i|\log^2|V|)$ operations are required for
  each $G_i$.  Hence, the total effort is
  $O(|V|+|E|+\log^2|V|\sum_i|E_i|)=O(|V|+|E|\log^2|V|)$.
\end{proof}

In order to illustrate the improved complexity for the construction of LRTs
of 2-BMGs, we implemented both the well-known triple-based approach,
i.e., the application of \texttt{BUILD} \cite{Aho:81} with the
informative triples defined in Eq.~(\ref{eq:informative-triples}) as
input, and the new approach of Alg.~\ref{alg:2-col-BMG}.  As input, we
used 2-BMGs that where randomly generated as follows: First, we simulate
random trees $T$ recursively, starting from a single vertex, by
attaching to a randomly chosen vertex $v$ either a single leaf if $v$ is
an inner vertex of $T$ or a pair of leaves if $v$ was a leaf. The
construction stops when the desired number of leaves is reached. Note
that the resulting tree is phylogenetic by construction.  Each leaf is
then colored by selecting at random one of the two colors. Finally, we
compute the 2-BMG $\G(T,\sigma)$ from each of the simulated leaf-colored
trees $(T,\sigma)$. 

\begin{figure}[t]
  \begin{center}
    \includegraphics[width=0.85\textwidth]{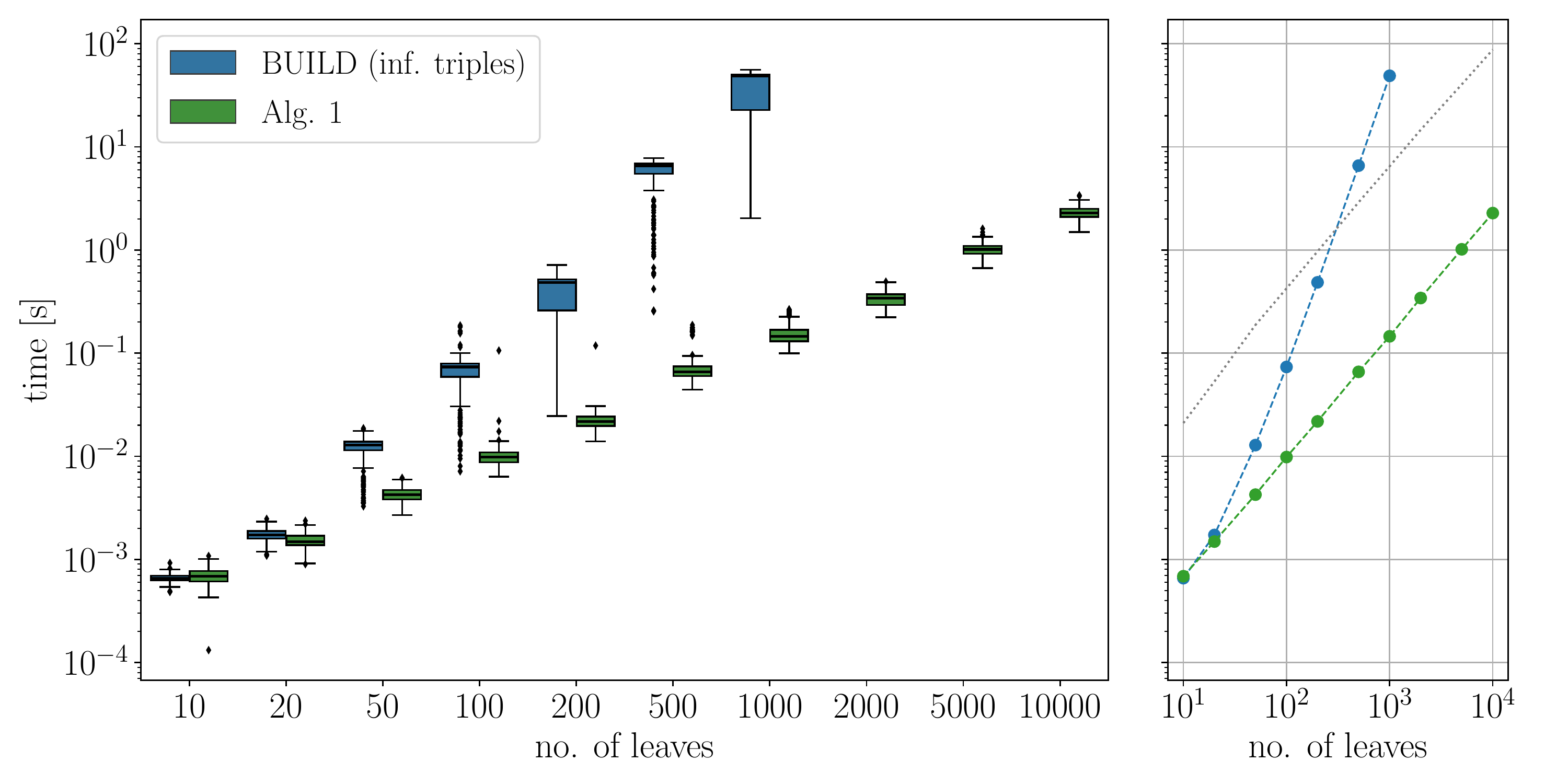}
  \end{center}
  \caption{Running time comparison of the general approach for constructing
    an LRT using \texttt{BUILD} (blue) vs.\ Alg.~\ref{alg:2-col-BMG}
    (green).  For each number of leaves, 200 2-BMGs where generated as
    described in the text. In the left panel, the median values are shown
    with logarithmic axes.  The additional dotted line indicates the median
    values of the size of the simulated BMGs, i.e.\ the number of arcs,
    scaled by a factor $10^{-3}$.  We did not compute the LRTs with the
    first method for instances with more than 1000 leaves because of the
    excessive computational cost.}
  \label{fig:runtime-boxplots}
\end{figure}

Both methods for the LRT computation were implemented in Python. Moreover,
we note that we did not implement the sophisticated dynamic data structures
used in the proof of Lemma~\ref{lem:connG}, but a rather na{\"i}ve
implementation of Alg.~\ref{alg:2-col-BMG}.  Nevertheless,
Fig.~\ref{fig:runtime-boxplots} shows a remarkable improvement of the
running time when compared to the general $O(|V|\,|E| \log^2 |V|)$ approach
for $\ell$-BMGs detailed in \cite{Geiss:19a}. Empirically, we observe that
the running time of Alg.~\ref{alg:2-col-BMG} indeed scales nearly linearly
with the number of edges. 

\section{Binary-explainable 2-BMGs}
\label{sect:be-2BMG}

Binary phylogenetic trees are of particular interest in practical
applications. Not every 2-BMG can be explained by a binary tree. The
subclass of \emph{binary-explainable ($\ell$-)BMG} are characterized
among all BMGs by the absence of single forbidden subgraph called
\emph{hourglass} \cite{Schaller:20p,Schaller:20x}, illustrated in
Fig.~\ref{fig:hourglass}.  In this section we briefly describe a
modification of Alg.~\ref{alg:2-col-BMG} that allows the efficient
recognition of binary-explainable 2-BMGs.
\begin{definition}
  An \emph{hourglass} in a properly vertex-colored graph $(\G,\sigma)$,
  denoted by $[xy \hourglass x'y']$, is a subgraph $(\G[Q],\sigma_{|Q})$
  induced by a set of four pairwise distinct vertices
  $Q=\{x, x', y, y'\}\subseteq V(\G)$ such that (i)
  $\sigma(x)=\sigma(x')\ne\sigma(y)=\sigma(y')$, (ii) $(x,y),(y,x)$ and 
  $(x'y'),(y',x')$ are bidirectional arcs
  in $\G$, (iii) $(x,y'),(y,x')\in E(\G)$, and (iv)
  $(y',x),(x',y)\notin E(\G)$.
\end{definition}

\begin{figure}[t]
  \begin{center}
    \includegraphics[width=0.45\linewidth]{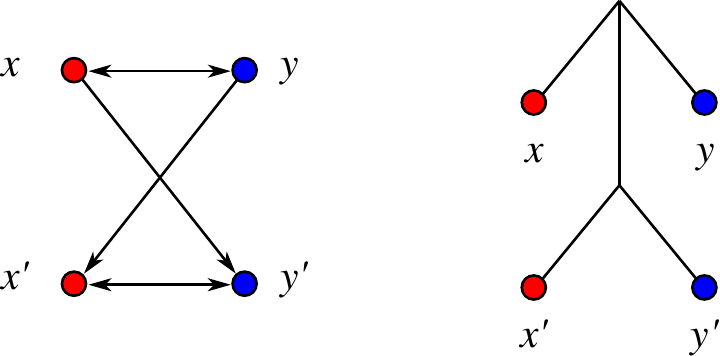}
  \end{center}
  \caption{The tree on the r.h.s.\ explains the hourglass graph on the
    l.h.s.}
  \label{fig:hourglass}
\end{figure}

A graph $(\G,\sigma)$ is called \emph{hourglass-free} if it does not contain 
an hourglass as an induced subgraph.
We summarize Lemma~31 and Prop.~8 in \cite{Schaller:20x} as
\begin{proposition}{}
  \label{prop:binary-iff-hourglass-free-iff-colors}
  For every BMG $(\G,\sigma)$, the following three statements are equivalent:
  \begin{enumerate}
    \item $(\G,\sigma)$ is binary-explainable.
    \item $(\G,\sigma)$ is hourglass-free.
    \item If $(T,\sigma)$ is a tree explaining $(\G,\sigma)$, then there is
      no vertex $u\in V^0(T)$ with three distinct children $v_1$, $v_2$,
      and $v_3$ and two distinct colors $r$ and $s$ satisfying
    \begin{enumerate}
      \item $r\in\sigma(L(T(v_1)))$, 
      $r,s\in\sigma(L(T(v_2)))$, 
      and $s\in\sigma(L(T(v_3)))$, and
      \item $s\notin\sigma(L(T(v_1)))$, and 
      $r\notin\sigma(L(T(v_3)))$.
    \end{enumerate}
  \end{enumerate}
\end{proposition}

The following Lemma shows that the third condition in
Prop.~\ref{prop:binary-iff-hourglass-free-iff-colors} can be translated to
a much simpler statement in terms of the support leaves of its LRT.
\begin{lemma}
  \label{lem:be-2BMG-complexity-}
  A 2-BMG $(\G,\sigma)$ contains an induced hourglass if and only if its LRT
  $(T,\sigma)$ contains an inner vertex $u$ such that $S_{u}$ contains support
  vertices of both colors and $V(\G(T(u))-S_{u})\ne\emptyset$.
\end{lemma}
\begin{proof}
  By Thm.~\ref{thm:2BMG-algo-works}, Alg.~\ref{alg:2-col-BMG} returns the LRT 
  $(T,\sigma)$ for $(\G,\sigma)$ if and only if $(\G,\sigma)$ is a 2-BMG.
  Hence, we assume in the following that the latter is satisfied.
  As a consequence of Prop.~\ref{prop:binary-iff-hourglass-free-iff-colors}
  and the fact that $(T,\sigma)$ explains $(\G,\sigma)$, we know that
  $(\G,\sigma)$ is binary-explainable if and only if there is no vertex
  $u\in V^0(T)$ with three distinct children $v_1$, $v_2$, and $v_3$ and
  two distinct colors $r$ and $s$ satisfying (a) $r\in\sigma(L(T(v_1)))$,
  $r,s\in\sigma(L(T(v_2)))$, and $s\in\sigma(L(T(v_3)))$, and (b)
  $s\notin\sigma(L(T(v_1)))$, and $r\notin\sigma(L(T(v_3)))$.

  First, suppose that $(\G,\sigma)$ contains an hourglass, i.e., by
  Prop.~\ref{prop:binary-iff-hourglass-free-iff-colors} there is a vertex
  $u\in V^0(T)$ with distinct children $v_1$, $v_2$, and $v_3$ and two
  distinct colors $r$ and $s$ satisfying (a) and (b).  Since $(\G,\sigma)$
  is 2-colored and $(T,\sigma)$ its LRT,
  Lemma~\ref{lem:single-color-is-leaf} together with
  $s\notin\sigma(L(T(v_1)))$ and $r\notin\sigma(L(T(v_3)))$ implies that
  $v_1$ of color $r$ and $v_2$ of color $s$, respectively, are both leaves.
  In particular, therefore, we know that $v_1, v_2\in S_{u}$ are support
  leaves.  By Lemma~\ref{lem:remove-support-leaves} and since
  $\G(T(u),\sigma\rst)$ is also a BMG, the connected components of
  $(\G(T(u))-S_{u},\sigma\rst)=(\G[L(T(u))]-S_{u},\sigma\rst)$ (cf.\
  Lemma~\ref{lem:subgraph}) are exactly the BMGs $\G(T(v),\sigma\rst)$ with
  $v\in\child(u)\setminus S_{u}$.  Together with the fact that
  $v_2\in V^0(T)$ as a consequence of $L(T(v_2))$ containing both colors
  $r$ and $s$, this implies that $(\G(T(u))-S_{u},\sigma\rst)$ is not the
  empty graph.

  Conversely, suppose there a vertex $u\in V^0(T)$ such that $S_{u}$
  contains support vertices $v_1$ and $v_3$ with distinct colors
  $\sigma(v_1)\ne\sigma(v_3)$ and $V(\G(T(u))-S_{u})\ne\emptyset$, i.e.,
  $u$ has a child $v_2\in \child(u)\setminus S_u$ that is not a support
  leaf and hence satisfies $v_2\in V^0(T)$.
  Lemma~\ref{lem:single-color-is-leaf} implies that $L(T(v_2))$ contains
  both colors since $v_2\in V^0(T)$. Hence, the three children $v_1$,
  $v_2$, and $v_3$ of $u$ satisfy conditions (a) and (b) of
  Prop.~\ref{prop:binary-iff-hourglass-free-iff-colors}(3), and thus
  $(\G,\sigma)$ contains an induced hourglass.
\end{proof}

\begin{corollary}
  \label{cor:be-2BMG-complexity}
  It can be checked in $O(|V| + |E|\log^2|V|)$ whether or not a properly
  2-colored graph $(\G,\sigma)$ is a binary-explainable BMG.
\end{corollary}
\begin{proof}
  Recall that there is a one-to-one correspondence between the recursion
  step in Alg.~\ref{alg:2-col-BMG} and the inner vertices $u\in V^0(T)$.
  As argued in the proof of Lemma~\ref{lem:connG}, every vertex appears at
  most twice in an umbrella set $U$. Therefore, it can be checked in
  $O(|V|)$ total time whether $S=S^{(2)}$ contains vertices of both
  colors. Since the vertex set of $\G_u-S_u$ is maintained in the dynamic
  graph HDT data structure, it can be checked in constant time for each $u$
  whether $\G_u-S_u$ is non-empty. The additional effort to check the
  condition of Lemma~\ref{lem:be-2BMG-complexity-} is therefore only
  $O(|V|)$.  Hence, we still require a total effort of
  $O(|V| + |E|\log^2|V|)$ (cf.\ Thm.~\ref{thm:2BMG-general-compl}).
\end{proof}
Cor.~\ref{cor:be-2BMG-complexity} improves the complexity for the decision 
whether a 2-BMG is binary-explainable as compared to the $O(|V|^3 \log^2 
|V|)$-time algorithm for (general) BMGs presented in \cite{Schaller:20p}.

\section{Concluding Remarks} 

We have shown here that 2-BMGs have a recursive structure that is reflected
in certain induced subgraphs that correspond to subtrees of the LRT. The
leaves connected directly to the root of a given subtree play a special
role as support vertices in the corresponding subgraph of the 2-BMG. Since
the support vertices of the root can be identified efficiently in a given
input graph, there is a recursive decomposition of $(\G,\sigma)$ that
directly yields the LRT. With the help of a dynamic data structure to
maintain connectedness information \cite{Holm:01}, this provides an
$O(|V|+|E|\log^2 |V|)$ algorithm to recognize both 2-BMGs and binary
explainable 2-BMGs and to construct the corresponding LRT. This provides
a considerable speed-up compared to the previously known
$O(|V||E| \log^2 |V|)$ and $O(|V|^3)$ algorithms. Empirically, we observe
a substantial speed-up even if simpler data structures are used to
implement Alg.~\ref{alg:2-col-BMG}.

Both the theoretical insights and Alg.~\ref{alg:2-col-BMG} itself have
potential applications to the analysis of gene families in computational
biology. Real-life data necessarily contain noise, and thus likely will
deviate from perfect BMGs, naturally leading to graph editing problems
for BMGs. Like many combinatorial problems in phylogenetics, these are
NP-complete \cite{Schaller:20y} and hence require approximation algorithms
and heuristics. The support leaves introduced here provide an avenue to a
new class of heuristics, conceptually distinct from approaches that
attempt to extract consistent subsets of triples from
$\mathscr{R}(\G,\sigma)$.

%% --------------------------------------------------------------------
%       Acknowledgements
%% --------------------------------------------------------------------

\section*{Acknowledgments}

This work was supported in part by the Austrian Federal Ministries BMK and
BMDW and the Province of Upper Austria in the frame of the COMET Programme
managed by FFG, and the \emph{Deutsche Forschungsgemeinschaft}. 

%% --------------------------------------------------------------------
%       Bibliography
%% --------------------------------------------------------------------
% Don't change: journal style want exactly this:

\clearpage
\bibliography{preprint-LRT-2BMG}
\bibliographystyle{abbrvurl}

\end{document}